\documentclass[11pt]{article}

\usepackage[left=1in, right=1in, top=1in, bottom=1in]{geometry}
\usepackage[group-separator={,}]{siunitx}
\usepackage{amsmath,amsfonts,amssymb, amsthm}
\usepackage{graphicx}
\usepackage{url}
\usepackage{mathtools}
\usepackage[textsize=tiny,textwidth=1cm,shadow]{todonotes}
\usepackage{thmtools}
\usepackage{enumitem}
\usepackage{pdflscape}
\usepackage{verbatim}
\usepackage{nccmath}
\usepackage[utf8]{inputenc}
\usepackage{array}
\definecolor{MyBlue}{rgb}{0.12, 0.12, 0.76}
\usepackage{algorithm}
\usepackage{algpseudocode}
\usepackage{algorithmicx}
\let\oldReturn\Return
\renewcommand{\Return}{\State\oldReturn}
\algtext*{EndWhile}% Remove "end while" text
\algtext*{EndIf}% Remove "end if" text
\algtext*{EndForAll}% Remove "end for" text
\algtext*{EndFor}% Remove "end for" text
\algtext*{EndFunction}% Remove "end function" text
\usepackage{pifont}% http://ctan.org/pkg/pifont
\usepackage{pgf}
\usepackage{tikz}
\usetikzlibrary{shadows,arrows,decorations,decorations.shapes,backgrounds,shapes,snakes,automata,fit,petri,shapes.multipart,calc,positioning,shapes.geometric,graphs,graphs.standard}
\makeatletter
\newcommand{\thickhline}{%
    \noalign {\ifnum 0=`}\fi \hrule height 1.4pt
    \futurelet \reserved@a \@xhline
}
\newcolumntype{"}{@{\hskip\tabcolsep\vrule width 1.4pt\hskip\tabcolsep}}
\makeatother
\allowdisplaybreaks[1] % let align break over multiple pages

\newtheorem{theorem}{Theorem}[section]
\newtheorem{lemma}{Lemma}[section]

\usepackage{natbib}
\usepackage{tablefootnote} 
\usepackage{epsfig}
\usepackage[colorlinks,allcolors=MyBlue]{hyperref}

\newlist{exlist}{enumerate}{1}
\setlist[exlist]{label=(\alph*)}

\newcommand{\eps}{\epsilon}

\newcommand{\prob}[2][]{\mathbf{Pr}\ifthenelse{\not\equal{}{#1}}{_{#1}}{}\!\left[#2\right]}
\newcommand{\expect}[2][]{\mathbf{E}\ifthenelse{\not\equal{}{#1}}{_{#1}}{}\!\left[#2\right]}

\newcommand{\prize}{v}
\newcommand{\prizes}{{\mathbf \prize}}
\newcommand{\prizei}[1][i]{{\prize_{#1}}}

\newcommand{\gdist}{D}
\newcommand{\gdists}{{\mathbf \gdist}}

\newcommand{\G}{\mathcal{G}}
\newcommand{\D}{\mathcal{D}}

\begin{document}

\title{Distributional Analysis\thanks{Chapter~8 of the book {\em Beyond the
      Worst-Case Analysis of Algorithms}~\citep{bwca}.}}
\author{Tim Roughgarden\thanks{Department of Computer Science,
    Columbia University.  Supported in part by NSF award
    CCF-1813188 and ARO award W911NF1910294.  Email: \texttt{tim.roughgarden@gmail.com.}}}

\maketitle

\begin{abstract} In {\em distributional} or {\em average-case analysis},
the goal is to design an algorithm with good-on-average performance
with respect to a
  specific probability distribution.  Distributional analysis can be
  useful for the study of general-purpose algorithms on
  ``non-pathological'' inputs, and for the design of specialized
  algorithms in applications in which there is detailed understanding of
  the relevant input distribution.  
For some problems, however, pure distributional analysis encourages
``overfitting'' an algorithmic solution to a particular distributional
assumption and a more robust analysis framework is called for.  This
chapter presents numerous examples of the pros and cons of distributional
analysis, highlighting some of its greatest hits while also setting
the stage for the hybrids of worst- and average-case analysis studied
in later chapters.
\end{abstract}

\section{Introduction}

Part~I of this book covered refinements of worst-case analysis which
do not impose any assumptions on the possible inputs.  Part~II
described several deterministic models of data, in which inputs to a
problem were restricted to those with properties that are plausibly
shared by all ``real-world'' inputs.  This chapter, and a
majority of the remaining chapters in the book, consider models that
include a {\em probability distribution} over inputs.  

\subsection{The Pros and Cons of Distributional Analysis}

In its purest form, the goal in distributional analysis is to analyze
the average performance of algorithms with respect to a specific 
input distribution, and perhaps also to design new algorithms that
perform particularly well for this distribution.
What do we hope to gain from such an analysis?
\begin{itemize}

\item In applications in which the input distribution is well understood
  (e.g., due to lots of recent and representative data),
  distributional analysis is well suited both to predict the
  performance of existing algorithms and to design algorithms
  specialized to the input distribution.

\item When there is a large gap between the empirical and worst-case
  performance of an algorithm, an input distribution can serve as a
  metaphor for ``non-pathological'' inputs.  Even if the input
  distribution cannot be taken literally, a good average-case
  bound is a plausibility argument for the algorithm's empirical
  performance.  The three examples in Section~\ref{s:classical} are in
  this spirit.

\item Optimizing performance with respect to a specific input
  distribution can lead to new algorithmic ideas that are
  useful much more broadly.  The examples in Sections~\ref{s:euclid}
  and~\ref{s:graph} have this flavor.

\end{itemize}
And what could go wrong?
\begin{itemize}

\item Carrying out an average-case analysis of an algorithm might be
  analytically tractable only for the simplest (and not necessarily
  realistic) input distributions.

\item Optimizing performance with respect to a specific input
  distribution can lead to ``overfitting,'' meaning algorithmic
  solutions that are overly reliant on the details of the
  distributional assumptions and have brittle performance guarantees
  (which may not hold if the distributional assumptions are violated).

\item Pursuing distribution-specific optimizations can distract from
  the pursuit of more robust and broadly useful algorithmic
  ideas.

\end{itemize}
This chapter has two goals.  The first is to celebrate a
few classical results in the average-case analysis of algorithms,
which constitute some of the earliest work on alternatives to
worst-case analysis.  Our coverage here is far from encyclopedic, with
the discussion confined to a sampling of relatively simple results for
well-known problems that contribute to the chapter's overarching
narrative.  The second goal is to examine critically such average-case
results, thereby motivating the more robust models of distributional
analysis outlined in Section~\ref{s:rda} and studied in detail later
in the book.

\subsection{An Optimal Stopping Problem}\label{ss:stop}

The pros and cons of distributional analysis are evident in a famous
example from optimal stopping theory, which is interesting in its own
right and also relevant to some of the random-order models described
in Chapter~11.  Consider a game with~$n$ stages.  Nonnegative prizes
arrive online, with~$v_i$ denoting the value of the prize that appears
in stage~$i$.  At each stage, an algorithm must decide between
accepting the current prize (which terminates the game) and proceeding
to the next stage after discarding it.  This involves a difficult
trade-off, between the risk of being too ambitious (and skipping over
what turns out to be the highest-value prize) and not ambitious enough
(settling for a modest-value prize instead of waiting for a better
one).

Suppose we posit specific distributions~$D_1,D_2,\ldots,D_n$, known in
advance to the algorithm designer, such that the value~$v_i$ of the
stage-$i$ prize is drawn independently from~$D_i$.  (The $D_i$'s may
or may not be identical.)  An algorithm learns the realization~$v_i$
of a prize value only at stage~$i$.  We can then speak about an {\em
  optimal} algorithm for the problem, meaning an online algorithm that
achieves the maximum-possible expected prize value, where the
expectation is with respect to the assumed 
distributions~$D_1,D_2,\ldots,D_n$.

The optimal algorithm for a given sequence of prize value
distributions is easy enough to specify, by working backward in time.
If an algorithm finds itself at stage~$n$ without having accepted a
prize, it should definitely accept the final prize.  (Recall all prizes have
nonnegative values.)  At an earlier stage~$i$, the algorithm should
accept the stage-$i$ prize if and only if $v_i$ is at least the
expected prize value obtained by the (inductively defined) optimal
strategy for stages~$i+1,i+2,\ldots,n$.

\subsection{Discussion}

The solution above illustrates the primary advantages of distributional
analysis: an unequivocal definition of an ``optimal'' algorithm, and
the possibility of a crisp characterization of such an algorithm
(as a function of the input distributions).

The disadvantages of average-case analysis are also on display, and
there are several reasons why one might reject this optimal
algorithm.
\begin{enumerate}

\item The algorithm takes the distributional assumptions literally and
  its description depends in a detailed way on the assumed
  distributions.  It is unclear how robust the optimality guarantee is
  to misspecifications of these distributions, or to a reordering of
  the distributions.

\item The algorithm is relatively complicated, in that it is defined
  by $n$ different parameters (one threshold for each stage).

\item The algorithm does not provide any qualitative advice about how
  to tackle similar problems (other than ``work backwards'').

\end{enumerate}
The third point is particularly relevant when studying a problem chosen
as a deliberate simplification of a ``real-world'' problem that is too
messy to analyze directly.  In this case, an optimal solution to the
simpler problem is useful only inasmuch as it suggests a plausibly
effective solution to the more general problem.

For our optimal stopping problem, could there be non-trivial guarantees
for simpler, more intuitive, and more robust algorithms?

\subsection{Threshold Rules and the Prophet Inequality}

Returning to the optimal stopping problem of Section~\ref{ss:stop},
a {\em threshold stopping rule} is defined by a single parameter, a
threshold~$t$.  The corresponding online algorithm accepts the first
prize~$i$ with value satisfying~$v_i \ge t$ (if any).  Such a rule is
clearly suboptimal, as it doesn't even necessarily accept the prize at
stage~$n$.  Nevertheless, the following {\em prophet inequality}
proves that there is a threshold strategy with an intuitive threshold
that performs approximately as well as a fully clairvoyant
prophet.\footnote{See Chapter~11 for an analogous result for a related
  problem, the {\em secretary problem}.}
\begin{theorem}[\citet{SC84}]\label{t:pi}
  For every sequence $\gdists=D_1,D_2,\ldots,D_n$ of independent prize value
  distributions, there is a threshold rule that guarantees expected
  reward at least
  $\tfrac{1}{2} \expect[\prizes \sim \gdists]{\max_i \prizei}$, where
  $\prizes$ denotes $(v_1,\ldots,v_n)$.
\end{theorem}

This guarantee holds, in particular, for the threshold~$t$ at which
there is a 50/50 chance that the rule accepts one of the~$n$ prizes.

\begin{proof}
Let $z^+$ denote $\max\{z,0\}$.  Consider a threshold strategy with
threshold~$t$ (to be chosen later).  
The plan is to prove a lower bound on the expected value of this
strategy and an upper bound on the expected value of a prophet such
that the two bounds are easy to compare.

What value does the $t$-threshold strategy obtain?  
Let $q(t)$ denote the probability of the failure mode where the
threshold strategy accepts no prize at all; in this case, it obtains 
zero value.  With the remaining probability $1-q(t)$, 
the rule obtains value at least~$t$.  To 
improve this lower bound, consider the case in which exactly one prize~$i$
satisfies $v_i \ge t$; then, the rule also gets ``extra credit'' of $v_i-t$
above and beyond its baseline value of $t$.\footnote{The difficulty when
two prizes~$i$ and~$j$ exceed the threshold is that this extra credit
is either $v_i-t$ or $v_j-t$ (whichever appeared earlier).  The proof
avoids reasoning about the ordering of the distributions by crediting
the rule only with the baseline value of~$t$ in this case.}

Formally, we can bound the expected value obtained by the
$t$-threshold strategy from below by
\begin{align}
\nonumber
&\ \qquad\qquad \qquad\qquad \qquad 
(1-q(t)) \cdot t~+ \\
\label{eq:t1}
&\ \sum_{i=1}^n \expect[\prizes]{v_i-t \,|\, v_i \ge t, v_j < t
  \,\,\forall j \neq i} \cdot \prob{v_i \ge t} 
 \cdot \prob{v_j < t \,\, \forall
  j \neq i}\\
\label{eq:t2}
 = &\
(1-q(t)) \cdot t + \sum_{i=1}^n \underbrace{\expect[\prizes]{v_i-t \,|\,
     v_i \ge t} \cdot
\prob{v_i \ge t}}_{=\expect{(v_i-t)^+}} \cdot
  \underbrace{\prob{v_j < t \,\, \forall   j \neq    i}}_{\ge q(t)} 
  \\
\label{eq:t3}
 \ge &\
(1-q(t)) \cdot t + q(t) \sum_{i=1}^n \expect[\prizes]{(v_i-t)^+} ,
\end{align}
where we use the independence of the $D_i$'s
in~\eqref{eq:t1} 
to factor
the two probability terms and 
in~\eqref{eq:t2} 
to drop the
conditioning on the event that $v_j < t$ for every $j \neq i$.
In~\eqref{eq:t3}, we use that $q(t) = \prob{v_j < t \,\, \forall j}
\le \prob{v_j < t \,\,\forall j \neq i}$.

Now we produce an upper bound on the prophet's expected value
$\expect[\prizes \sim \gdists]{\max_i v_i}$ that is easy to compare to~\eqref{eq:t3}.
The expression $\expect[\prizes]{\max_i v_i}$ doesn't reference the
strategy's threshold $t$, so we add and subtract it to derive
\begin{eqnarray}
\nonumber
\expect[\prizes]{\max_{i=1}^n v_i} & = & 
\expect[\prizes]{t + \max_{i=1}^n (v_i - t)}\\
\nonumber
& \le & t + \expect[\prizes]{\max_{i=1}^n (v_i - t)^+}\\
\label{eq:ub}
& \le & t + \sum_{i=1}^n \expect[\prizes]{(v_i - t)^+}.
\end{eqnarray}
Comparing~\eqref{eq:t3} and~\eqref{eq:ub}, we can complete the proof
by setting~$t$ so that $q(t) = \tfrac{1}{2}$, with a 50/50 chance of
accepting a prize.\footnote{If there is no such~$t$
  because of point masses in the $D_i$'s, then a minor extension of
  the argument yields the same result (Exercise~\ref{exer:prophetmass}).}
\end{proof}

The drawback of this threshold rule relative to the optimal online
algorithm is clear: it does not guarantee as much expected value.  
Nonetheless, this solution possesses several attractive properties
that are not satisfied by the optimal algorithm:
\begin{enumerate}

\item The threshold rule recommended by Theorem~\ref{t:pi} depends on
  the prize value distributions~$D_1,D_2,\ldots,D_n$ only inasmuch as
  it depends on the number~$t$ for which there is a 50/50 probability
  that at least one realized value exceeds~$t$.  For example,
  reordering the distributions arbitrarily does not change the
  recommended threshold rule.

\item A threshold rule is simple in that it is defined by only one
  parameter.  Intuitively, a single-parameter rule is less prone to
  ``overfitting'' to the assumed distributions than a more highly
  parameterized algorithm like the ($n$-parameter) optimal
  algorithm.\footnote{See Chapter~29 on data-driven 
    algorithm design for a formalization of this intuition.}

\item Theorem~\ref{t:pi} gives flexible qualitative advice about how
  to approach such problems: Start with threshold rules, and don't be too
  risk-averse (i.e., choose an ambitious enough threshold that
  receiving no prize is a distinct possibility).

\end{enumerate}

\section{Average-Case Justifications of Classical Algorithms}\label{s:classical}

Distributional assumptions can guide the design of algorithms, as with
the optimal stopping problem introduced in Section~\ref{ss:stop}.
Distributional analysis can also be used to analyze a general-purpose
algorithm, with the goal of explaining mathematically why its
empirical performance is much better than its worst-case performance.
In these applications, the assumed probability distribution over
inputs should not be taken literally; rather, it serves as a metaphor
for ``real-world'' or ``non-pathological'' inputs.  This section
gives the flavor of work along these lines by describing one result
for each of three classical problems:
sorting, hashing, and bin packing.

\subsection{QuickSort}

Recall the QuickSort algorithm from undergraduate algorithms which,
given an array of $n$ elements from a totally ordered set, works as
follows:
\begin{itemize}

\item Designate one the~$n$ array entries as a ``pivot'' element.

\item Partition the input array around the pivot element~$p$, meaning
  rearrange the array entries so that all entries less than~$p$ appear
  before~$p$ in the array and all entries greater than~$p$ appear
  after~$p$ in the array.

\item Recursively sort the subarray comprising the elements less than~$p$.

\item Recursively sort the subarray comprising the elements greater than~$p$.

\end{itemize}
The second step of the algorithm is easy to implement in~$\Theta(n)$ time.
There are many ways to choose the pivot element, and the running time
of the algorithm varies between~$\Theta(n \log n)$ and~$\Theta(n^2)$,
depending on these choices.\footnote{In the best-case scenario,
every pivot element is the median element of the subarray, leading to
the recurrence
$T(n) = 2 T(\tfrac{n}{2}) + \Theta(n)$ with solution~$\Theta(n \log n)$.
In the worst-case scenario,
every pivot element is the minimum or maximum element of the subarray,
leading to the recurrence
$T(n) = T(n-1) + \Theta(n)$ with solution~$\Theta(n^2)$.}
One way to enforce the best-case scenario is to explicitly compute the
median element and use it as the pivot.  A simpler and more practical
solution is to choose the pivot element uniformly at random; most of
the time, it will be close enough to the median that both recursive
calls are on significantly smaller inputs.  A still simpler solution,
which is common in practice, is to always use the first array element
as the pivot element.  This deterministic version of QuickSort
runs in~$\Theta(n^2)$ time on already-sorted arrays, but empirically
its running time is~$\Theta(n \log n)$ on almost all other inputs.
One way to formalize this observation is to analyze the algorithm's
expected running time on a random input.  As a comparison-based
sorting algorithm, the running time of QuickSort depends only on the
relative order of the array entries, so we can assume without loss of
generality that the input is a permutation of~$\{1,2,\ldots,n\}$ and
identify a ``random input'' with a random permutation.  
With any of the standard implementations of the partitioning
subroutine, the average-case running time of this deterministic
QuickSort algorithm is at most a constant factor larger than its
best-case running time.

\begin{theorem}[\citet{hoare}]\label{t:qs}
The expected running time of the deterministic QuickSort algorithm on
a random permutation of~$\{1,2,\ldots,n\}$ is~$O(n \log n)$.
\end{theorem}

\begin{proof}
We sketch one of the standard proofs.  Assume that the partitioning
subroutine only makes comparisons that involve the pivot element; this is the
case for all of the textbook implementations.  Each recursive call is
given a subarray consisting of the elements from some
interval~$\{i,i+1,\ldots,j\}$; conditioned on this interval, the
relative order of its elements in the subarray is uniformly random.

Fix elements $i$ and $j$ with $i < j$.  These elements are passed to
the same sequence of recursive calls (along
with~$i+1,i+2,\ldots,j-1$), up to the first call in which an element
from~$\{i,i+1,\ldots,j\}$ is chosen as a pivot element.  At this
point, $i$ and $j$ are either compared to each other (if $i$ or $j$
was the chosen pivot) or not (otherwise); in any case, they are never
compared to each other again in the future.  With all subarray
orderings equally likely, the probability that $i$ and $j$ are
compared is exactly~$\tfrac{2}{j-i+1}$.  By the linearity of
expectation, the expected total number of comparisons is then
$\sum_{i=1}^{n-1} \sum_{j=i+1}^n \tfrac{2}{j-i+1} = O(n \log n)$, and
the expected running time of the algorithm is at most a constant
factor larger.
\end{proof}

\subsection{Linear Probing}\label{ss:knuth}

A hash table is a data structure that supports fast insertions and
lookups.  Under the hood, most hash table implementations maintain an
array~$A$ of some length~$n$ and use a hash function~$h$ to map each
inserted object~$x$ to an array entry~$h(x) \in \{1,2,\ldots,n\}$.  A
fundamental issue in hash table design is how to resolve collisions,
meaning pairs~$x,y$ of distinct inserted objects for which
$h(x)=h(y)$.  {\em Linear probing} is a specific way of resolving
collisions:
\begin{enumerate}

\item Initially, all entries of~$A$ are empty.

\item Store a newly inserted object~$x$ in the first empty entry in
  the sequence $A[h(x)]$, $A[h(x)+1], A[h(x)+2], \ldots$, wrapping
  around to the beginning of the array, if necessary.

\item To search for an object~$x$, scan the entries
  $A[h(x)], A[h(x)+1], A[h(x)+2], \ldots$ until encountering $x$ (a
  successful search) or an empty slot (an unsuccessful search),
  wrapping  around to the beginning of the array, if necessary.

\end{enumerate}
That is, the hash function indicates the starting position for an
insertion or lookup operation, and the operation scans to the right
until it finds the desired object or an empty position.
The running time of an insertion or lookup operation
is proportional to the length of this scan.

The bigger the fraction~$\alpha$ of the hash table that is occupied
(called its {\em load}),
the fewer empty array entries and the longer the
scans.  To get calibrated, imagine searching for an empty
array entry using independent and uniformly random probes.  The number
of attempts until a success is then a geometric random variable with
success probability $1-\alpha$, which has expected value
$\tfrac{1}{1-\alpha}$.  With linear probing, however, objects tend to
clump together in consecutive slots, resulting in slower operation
times.  How much slower?

Non-trivial mathematical guarantees for hash tables are possible only
under assumptions that rule out data sets that are pathologically
correlated with the table's hash function; for this reason, hash
tables have long constituted one of the killer applications of
average-case analysis.  Common assumptions include asserting some
amount of randomness in the data (as in average-case analysis), in the
choice of hash function (as in randomized algorithms), or both (as in
Chapter~26).  For example, assuming that the data and hash function
are such that every hash value~$h(x)$ is an independent and uniform
draw from~$\{1,2,\ldots,n\}$, the expected time of insertions and
lookups scales with $\tfrac{1}{(1-\alpha)^2}$.\footnote{This result
  played an important role in the genesis of the mathematical analysis
  of algorithms.  Donald E.\ Knuth, its discoverer, wrote: ``I first
  formulated the following derivation in 1962\ldots Ever since that
  day, the analysis of algorithms has in fact been one of the major
  themes in my life.''}

\subsection{Bin Packing}\label{ss:bp}

The {\em bin packing} problem played a central role in the early
development of the average-case analysis of algorithms; this section
presents one representative result.\footnote{See Chapter~11 for
  an analysis of bin packing heuristics in random-order models.}
Here, the average-case analysis is
of the solution quality output by a heuristic (as with the prophet
inequality), not its running time (unlike our QuickSort and linear
probing examples).

In the bin packing problem, the input is~$n$ items with
sizes~$s_1,s_2,\ldots,s_n \in [0,1]$.  Feasible solutions correspond
to ways of partitioning the items into bins so that the sum of the
sizes in each bin is at most~1.  The objective is to minimize the number of
bins used.  This problem is $NP$-hard, so every polynomial-time
algorithm produces suboptimal solutions in some cases (assuming
$P \neq NP$).

Many practical bin packing heuristics have been studied extensively
from both worst-case and average-case viewpoints.  One
example is the {\em first-fit decreasing (FFD)} algorithm:
\begin{itemize}

\item Sort and reindex the items so that $s_1 \ge s_2 \ge \cdots s_n$.

\item For~$i=1,2,\ldots,n$:

\begin{itemize}

\item If there is an existing bin with room for item~$i$ (i.e., with
  current total size at most~$1-s_i$), add~$i$ to the first such bin.

\item Otherwise, start a new bin and add~$i$ to it.

\end{itemize}

\end{itemize}
For example, consider an input consisting of~6 items with
size~$\tfrac{1}{2}+\eps$, 6 items with size $\tfrac{1}{4}+2\eps$, 6
jobs with size $\tfrac{1}{4}+\eps$, and 12 items with
size~$\tfrac{1}{4}-2\eps$.  The FFD algorithm uses~11 bins while an
optimal solution packs them perfectly into~9 bins
(Exercise~\ref{exer:ffd}).  Duplicating this set of~30 jobs as many
times as necessary shows that there are arbitrarily large inputs for
which the FFD algorithm uses~$\tfrac{11}{9}$ times as many bins as an
optimal solution.  Conversely, the FFD algorithm never uses more than
$\tfrac{11}{9}$ times the minimum-possible number of bins plus an
additive constant (see the Notes for details).

The factor of~$\tfrac{11}{9} \approx 1.22$ is quite good as worst-case
approximation ratios go, but empirically the FFD algorithm usually
produces a solution that is extremely close to optimal.
One approach to better theoretical bounds is distributional analysis.
For bin-packing algorithms, the natural starting point
is 
the case in which item sizes are independent draws from the uniform
distribution on~$[0,1]$.  Under this (strong) assumption, 
the FFD algorithm is near-optimal in a strong sense.

\begin{theorem}[\cite{F80}]\label{t:f80}
  For every $\eps > 0$, for~$n$ items with sizes distributed
  independently and uniformly in~$[0,1]$, with probability~$1-o(1)$ as
  $n \rightarrow \infty$, the FFD algorithm uses less than~$(1+\eps)$
  times as many bins as an optimal solution.
\end{theorem}
In other words, the typical approximation ratio of the FFD algorithm
tends to~1 as the input size grows large.

We outline a two-step proof of Theorem~\ref{t:f80}.  
The first step shows that the guarantee holds for a less natural
algorithm that we call the {\em truncate and
  match (TM)} algorithm.  The second step shows that the FFD algorithm
never uses more bins than the TM algorithm.

The truncate and match algorithm works as follows:
\begin{itemize}

\item Pack every item with size at least $1-\tfrac{2}{n^{1/4}}$ in its
  own bin.\footnote{For clarity, we omit ceilings and floors.  See
    Exercise~\ref{exer:f80a} for the motivation behind this size cutoff.}

\item Sort and reindex the remaining~$k$ items so that $s_1 \ge s_2
  \ge \cdots \ge s_k$.  (Assume for simplicity that~$k$ is even.)

\item For each $i=1,\ldots,k/2$, put items~$i$ and $k-i+1$ into a
  common bin if possible; otherwise, put them in separate bins.

\end{itemize}

To explain the intuition behind the TM algorithm, consider the
expected order
statistics (i.e., expected minimum, expected second-minimum, etc.)
of~$n$ independent samples from the uniform distribution on~$[0,1]$.
It can be shown that these split~$[0,1]$ evenly into~$n+1$
subintervals; the expected minimum is~$\tfrac{1}{n+1}$, the expected
second-minimum~$\tfrac{2}{n+1}$, and so on.
Thus at least in an expected sense, the first and last items together
should fill up a bin exactly, as should the second and second-to-last
items, and so on.  Moreover, as $n$ grows large, the difference
between the realized order statistics and their expectations should
become small.  Setting aside a small number of the largest items in
the first step then corrects for any (small) deviations from these
expectations with negligible additional cost.  See
Exercise~\ref{exer:f80a} for details.

We leave the second step of the proof of Theorem~\ref{t:f80} as
Exercise~\ref{exer:f80b}.

\begin{lemma}\label{l:f80}
For every bin packing input, the FFD algorithm uses at most as many
bins as the TM algorithm.
\end{lemma}

The description of the general-purpose FFD algorithm is not tailored
to a distributional assumption, but the proof of Theorem~\ref{t:f80}
is fairly specific to uniform-type distributions.  This is emblematic
of one of the drawbacks of average-case analysis: Often, it is
analytically tractable only under quite specific distributional
assumptions.

\section{Good-on-Average Algorithms for Euclidean Problems}\label{s:euclid}

Another classical application domain for average-case analysis is in
computational geometry, with the input comprising random points from
some subset of Euclidean space.  We highlight two representative
results for fundamental problems in two dimensions, one concerning the
running time of an always-correct convex hull algorithm and one about
the solution quality of an efficient heuristic for the $NP$-hard Traveling
Salesman Problem.  

\subsection{2D Convex Hull}\label{ss:ch}

A typical textbook on computational geometry begins with the
{\em 2D convex hull} problem.  The input consists of a set~$S$ of~$n$
points in the plane (in the unit square $[0,1] \times [0,1]$, say) and
the goal is to report, in sorted order, the points of~$S$ that lie on
the convex
  hull of~$S$.\footnote{Recall that the {\em convex hull} of a set of
    points is the smallest convex set containing them, or equivalently
    the set of all convex combinations of points from~$S$.  In
    two dimensions, imagine the points as nails in a board, and
    the convex hull as a taut rubber band that encloses them.}
There are several algorithms that solve the 2D convex hull problem
in~$\Theta(n \log n)$ time.  Can we do better---perhaps even linear
time---when the points are drawn from a distribution, such as the
uniform distribution on the square?

\begin{theorem}[\citet{BS78}]\label{t:ch}
There is an algorithm that solves the 2D convex hull problem in
expected~$O(n)$ time for $n$ points drawn independently and uniformly
from the unit square.
\end{theorem}

The algorithm is a simple divide-and-conquer algorithm.
Given points $S=\{p_1,p_2,\ldots,p_n\}$ drawn independently and
uniformly from the plane:
\begin{itemize}

\item If the input~$S$ contains at most~5 points, compute the convex
  hull by brute force.  Return the points of~$S$ on the convex hull,
  sorted by~$x$-coordinate.

\item Otherwise, let~$S_1 = \{p_1,\ldots,p_{n/2}\}$
  and~$S_2=\{p_{(n/2)+1},\ldots,p_n\}$ denote the first and second
  halves of~$S$.  (Assume for simplicity that~$n$ is even.)

\item Recursively compute the convex hull~$C_1$ of~$S_1$, with its
  points sorted by $x$-coordinate.

\item Recursively compute the convex hull~$C_2$ of~$S_2$, with its
  points sorted by $x$-coordinate.

\item Merge~$C_1$ and~$C_2$ into the convex hull~$C$ of~$S$.
  Return~$C$, with the points of~$C$ sorted by $x$-coordinate.
  
\end{itemize}
For every set~$S$ and partition of~$S$ into~$S_1$ and~$S_2$, every
point on the convex hull of~$S$ is on the convex hull of either~$S_1$
or~$S_2$.  Correctness of the algorithm follows immediately.  The last
step is easy to implement in time linear in $|C_1|+|C_2|$; see
Exercise~\ref{exer:graham}.  Because the subproblems~$S_1$ and~$S_2$
are themselves uniformly random points from the unit square (with the
sorting occurring only after the recursive computation completes), the
expected running time of the algorithm is governed by the recurrence
\[
T(n) \le 2 \cdot T(\tfrac{n}{2}) + O(\expect{|C_1|+|C_2|}).
\]
Theorem~\ref{t:ch} follows immediately from this recurrence and the
following combinatorial bound.
\begin{lemma}[\citet{RS63}]\label{l:ch}
The expected size of the convex hull of~$n$ points drawn independently
and uniformly from the unit square is~$O(\log n)$.
\end{lemma}

\begin{proof}
Imagine drawing the input points in two phases, with $\tfrac{n}{2}$
points~$S_i$ drawn in phase~$i$ for $i=1,2$.
An elementary argument shows that the convex hull of the points in~$S_1$
occupies, in expectation, at least a $1-O(\tfrac{\log n}{n})$ fraction
of the unit square (Exercise~\ref{exer:ch}).
Each point of the second phase thus lies in the
interior of the convex hull of~$S_1$ (and hence of~$S_1 \cup S_2$)
except with probability~$O(\tfrac{\log n}{n})$, so the expected number of
points from~$S_2$ on the convex hull of~$S_1 \cup S_2$ is~$O(\log
n)$.  By symmetry, the same is true of~$S_1$.
\end{proof}

\subsection{The Traveling Salesman Problem in the Plane}\label{ss:tsp}

In the {\em Traveling Salesman Problem (TSP)}, the input consists of
$n$ points and distances between them, and the goal is to compute a
tour of the points (visiting each point once and returning to the
starting point) with the minimum-possible total length.
In {\em Euclidean} TSP, the points lie in Euclidean space and all
distances are straight-line distances.  This problem is $NP$-hard, even
in two dimensions.  The main result of this section is analogous to
Theorem~\ref{t:f80} in Section~\ref{ss:bp} for the bin packing
problem---a polynomial-time algorithm that, when the input points are
drawn independently and uniformly from the unit square,
has approximation ratio tending to~1 (with high probability) as $n$
tends to infinity.

The algorithm, which we call the {\em Stitch} algorithm, works as
follows:
\begin{itemize}

\item Divide the unit square evenly into $s = \tfrac{n}{\ln n}$
  subsquares, each with side length~$\sqrt{(\ln n)/n}$.\footnote{Again, we
    ignore ceilings and floors.}

\item For each subsquare~$i=1,2,\ldots,s$, containing
  the points~$P_i$:

\begin{itemize}

\item If $|P_i| \le 6 \log_2 n$, compute the
  optimal tour~$T_i$ of~$P_i$ using dynamic
  programming.\footnote{Given~$k$ points, label
    them~$\{1,2,\ldots,k\}$.  There is one subproblem for each
    subset~$S$ of points and point~$j \in S$, whose solution is the
    minimum-length path that starts at the point~1, ends at the
    point~$j$, and visits every point of~$S$ exactly once.  Each of
    the~$O(k2^k)$ subproblems can be solved in~$O(k)$ time by trying all
    possibilities for the final hop of the optimal path.
    When~$k=O(\log n)$, this running time of~$O(k^22^k)$ is polynomial
    in~$n$.}

\item Otherwise, return an arbitrary tour~$T_i$ of~$P_i$.

\end{itemize}

\item Choose an arbitrary representative point from each non-empty
  set~$P_i$, and let~$R$ denote the set of representatives.

\item Construct a tour~$T_0$ of~$R$ by visiting points from
  left-to-right in the
  bottommost row of subsquares, right-to-left in the second-to-bottom
  row, and so on, returning to the starting point after visiting all
  the points in the topmost row.

\item Shortcut the union of the subtours $\cup_{i=0}^{s} T_i$ to a
  single tour~$T$ of all~$n$ points, and return~$T$.\footnote{The
    union of the $s+1$ subtours can be viewed as a connected Eulerian
    graph, which then admits a closed Eulerian walk (using every edge
    of the graph exactly once).  This walk can be transformed to a
    tour of the points with only smaller length by skipping repeated
    visits to a point.\label{foot:euler}}

\end{itemize}
This algorithm runs in polynomial time with probability~1 and returns
a tour of the input points.  As for the approximation guarantee:

\begin{theorem}[\cite{K77}]\label{t:k77}
  For every $\eps > 0$, for~$n$ points distributed
  independently and uniformly in the unit square, with
  probability~$1-o(1)$ as
  $n \rightarrow \infty$, the Stitch algorithm returns a tour with
  total length less than~$(1+\eps)$
  times that of an optimal tour.
\end{theorem}

Proving Theorem~\ref{t:k77} requires understanding the typical
length of an optimal tour of random points in the unit square and
then bounding from above the difference between the lengths of the tour
returned by the Stitch algorithm and of the optimal tour.  The first
step is not difficult (Exercise~\ref{exer:k77}).

\begin{lemma}\label{l:k77a}
There is a constant $c_1 > 0$ such that, with probability~$1-o(1)$ as $n
\rightarrow \infty$, the length of an optimal tour of $n$ points drawn
independently and uniformly from the unit square is at least $c_1 
\sqrt{n}$.
\end{lemma}
Lemma~\ref{l:k77a} implies that proving Theorem~\ref{t:k77} reduces to
showing that (with high probability) the difference between the lengths
of Stitch's tour and the optimal tour is~$o(\sqrt{n})$.

For the second step, we start with a simple consequence of the
Chernoff bound (see Exercise~\ref{exer:k77b}).
\begin{lemma}\label{l:k77b}
In the Stitch algorithm, with probability~$1-o(1)$ as $n \rightarrow
\infty$, every subsquare contains at most~$6 \log_2 n$ points.
\end{lemma}

It is also easy to bound the length of the tour~$T_0$ of the
representative points~$R$ in the Stitch algorithm (see
Exercise~\ref{exer:k77c}). 
\begin{lemma}\label{l:k77c}
There is a constant $c_2$ such that, for every input, the length of the
tour~$T_0$ in the Stitch algorithm is at most
\[
c_2 \cdot \sqrt{s} = c_2 \cdot \sqrt{\frac{n}{\ln n}}.
\]
\end{lemma}

The key lemma states that an optimal tour can be massaged into
subtours for all of the subsquares without much additional cost.
\begin{lemma}\label{l:patch}
Let~$T^*$ denote an optimal tour of the~$n$ input points, and
let~$L_i$ denote the length of the portion of~$T^*$ that lies within
the subsquare~$i \in \{1,2,\ldots,s\}$ defined by the Stitch algorithm.
For every subsquare~$i=1,2,\ldots,s$,
there exists a tour of the points~$P_i$ in the subsquare of length at
most
\begin{equation}\label{eq:k77}
L_i + 6\sqrt{\frac{\ln n}{n}}.
\end{equation}
\end{lemma}
The key point in Lemma~\ref{l:patch} is that the upper bound
in~\eqref{eq:k77} depends only on the size of the square, and
not on the number of times that the optimal tour~$T^*$ crosses its
boundaries.

Before proving Lemma~\ref{l:patch}, we observe that
Lemmas~\ref{l:k77a}--\ref{l:patch} easily imply Theorem~\ref{t:k77}.
Indeed, with high probability:
\begin{enumerate} 

\item The optimal tour has length~$L^* \ge c_1 \sqrt{n}$.

\item Every subsquare in the Stitch algorithm contains at most $6 \ln
  n$ points, and hence the algorithm computes an optimal tour of the
  points in each subsquare (with length at most~\eqref{eq:k77}).

\item Thus, recalling that~$s = \tfrac{n}{\ln n}$, the total length of
  Stitch's tour is at most
\[
\sum_{i=1}^s \left( L_i +
6\sqrt{\frac{\ln n}{n}}
 \right)
+
c_2 \cdot \sqrt{\frac{n}{\ln n}} = L^* + O\left( \sqrt{\frac{n}{\ln n}}
\right) = (1+o(1)) \cdot L^*.
\]

\end{enumerate}

Finally, we prove Lemma~\ref{l:patch}.

\begin{proof} (Lemma~\ref{l:patch})
Fix a subsquare~$i$ with a non-empty set~$P_i$ of points.  The optimal
tour~$T^*$ visits every point in~$P_i$ while crossing the boundary of
the subsquare an even number~$2t$ of times; denote these
crossing points by~$Q_i=\{y_1,y_2,\ldots,y_{2t}\}$, indexed in clockwise order
around the subsquare's perimeter (starting from the lower left
corner).  Now form a connected Eulerian multi-graph~$G=(V,E)$ 
with vertices~$V = P_i \cup Q_i$ by adding the following edges:
\begin{itemize}

\item Add the portions of~$T^*$ that lie inside the subsquare (giving
  points of~$P_i$ a degree of~2 and points of~$Q_i$ a degree of~1).

\item Let~$M_1$ (respectively, $M_2$) denote the perfect matching
  of~$Q_i$ that matches each $y_j$ with $j$ odd (respectively, with $j$
  even) to $y_{j+1}$.  (In~$M_2$, $y_{2t}$ is matched with~$y_1$.)
Add two copies of the cheaper matching to the edge set~$E$ and one
copy of the more expensive matching (boosting the degree of points
of~$Q_i$ to~4 while also ensuring connectivity).

\end{itemize}
The total length of the edges contributed by the first ingredient
is~$L_i$.  The total length of the edges in $M_1 \cup M_2$ is at most
the perimeter of the subsquare, which is $4\sqrt{\tfrac{\ln n}{n}}$.
The second copy of the cheaper matching adds at
most~$2\sqrt{\tfrac{\ln n}{n}}$ to the total length of the edges
in~$G$.  
As in footnote~12, 
because~$G$ is connected and
Eulerian, we can extract from it a tour of $P_i \cup Q_i$ (and hence
of~$P_i$) that has total length at most that of the edges of~$G$, which
is at most $L_i + 6\sqrt{\tfrac{\ln n}{n}}$.
\end{proof}

\subsection{Discussion}\label{ss:disc}

To what extent are the
two divide-and-conquer algorithms of this section 
tailored to the distributional assumption that the input points are
drawn independently and uniformly at random from the unit square?  For
the convex hull algorithm in Section~\ref{ss:ch}, the consequence of
an incorrect distributional assumption is mild; its worst-case running
time is governed by the recurrence~$T(n) \le 2T(\tfrac{n}{2})+O(n)$
and hence is~$O(n \log n)$, which is close to linear.  Also, analogs
of Lemma~\ref{l:ch} (and hence Theorem~\ref{t:ch}) can be shown to
hold for a number of other distributions.

The Stitch algorithm in Section~\ref{ss:tsp}, with its fixed
dissection of the unit square into equal-size subsquares, may appear
hopelessly tied to the assumption of a uniform distribution.  But
minor modifications to it result in more robust algorithms, for
example by using an {\em adaptive} dissection, which recursively
divides each square along either the median $x$-coordinate or the
median $y$-coordinate of the points in the square.  Indeed, this idea
paved the way for later algorithms that obtained polynomial-time
approximation schemes (i.e., $(1+\eps)$-approximations for arbitrarily
small constant $\eps$) even for the {\em worst-case} version of 
Euclidean TSP (see the Notes).

Zooming out, our discussion of these two examples touches on one of the
biggest risks of
average-case analysis: distributional assumptions can lead to
algorithms that are unduly tailored to the assumptions.  
On the other hand, even when this is the case, the high-level ideas behind
the algorithms can prove useful much more broadly.

\section{Random Graphs and Planted Models}\label{s:graph}

Most of our average-case models so far concern random numerical data.
This section studies random combinatorial structures, and specifically
different probability distributions over graphs.

\subsection{Erd\H{o}s-R\'{e}nyi Random Graphs}\label{ss:er}

This section reviews the most well-studied model of random graphs, the
{\em Erd\H{o}s-R\'enyi} random graph model.  This model is a family
$\{ \G_{n,p} \}$ of distributions, indexed by the number~$n$ of
vertices and the edge density~$p$.  A sample from the distribution
$\G_{n,p}$ is a graph~$G=(V,E)$ with $|V|=n$ and each of the
$\binom{n}{2}$ possible edges present independently with
probability~$p$.  The special case of $p=\tfrac{1}{2}$ is the uniform
distribution over all $n$-vertex graphs.  This is an example of an
``oblivious random model,'' meaning that it is defined independently of
any particular optimization problem.

The assumption of uniformly random data may have felt like cheating
already in our previous examples,
but it is
particularly problematic for many computational problems on graphs.
Not only is this distributional assumption extremely specific, it also
fails to meaningfully differentiate between different
algorithms.\footnote{It also fails to replicate the statistical
  properties commonly observed in ``real-world'' graphs; see
  Chapter~28 for further discussion.}
We illustrate this point with two problems that are discussed at
length in Chapters~9 and~10.

%\vspace{-.5\baselineskip}
\paragraph{Example: Minimum bisection.}
In the {\em graph bisection} problem,
the input is an undirected graph $G=(V,E)$ with an even
number of vertices, and the goal is to identify a bisection
(i.e., a partition of~$V$ into two equal-size groups)
with the fewest number of crossing edges.

To see why this problem is algorithmically uninteresting in the
Erd\H{o}s-R\'enyi random graph model, take $p=\tfrac{1}{2}$ and let~$n$
tend to infinity.  In a random sample from~$\G_{n,p}$,
for every bisection~$(S,\bar{S})$ of the set~$V$ of~$n$ vertices,
the expected number of edges of~$E$ crossing it is $\tfrac{n^2}{8}$.  
A straightforward application of the Chernoff bound
shows that, with probability~$1-o(1)$ as $n \rightarrow \infty$, the
number of edges crossing {\em every} bisection is
$(1 \pm o(1)) \cdot \tfrac{n^2}{8}$ (Exercise~\ref{exer:bisection}).
Thus even an algorithm that
computes a {\em maximum} bisection is an almost optimal algorithm for computing a minimum bisection!

%\vspace{-\baselineskip}
\paragraph{Example: Maximum clique.}
In the {\em maximum clique} problem, the goal (given an undirected
graph) is to identify the largest subset of vertices that are
mutually adjacent.  In a random graph in the $\G_{n,1/2}$
model, the size of the maximum clique is very likely to be $\approx
2 \log_2 n$.\footnote{In fact, the size of the maximum clique turns
  out to be incredibly concentrated; see the Notes.}
To see heuristically why this is true, note that for an integer $k$, the
expected number of cliques on $k$ vertices in a random graph of
$\G_{n,1/2}$ is exactly
$$
\binom{n}{k} 2^{-\binom{k}{2}} \approx n^k 2^{-k^2/2},
$$
which is~1 precisely when $k = 2 \log_2 n$.  That is, $2 \log_2
n$ is approximately the largest~$k$ for which we expect to see at least one
$k$-clique.

On the other hand, while there are several polynomial-time algorithms
(including the obvious greedy algorithm) that compute, with high
probability, a clique of size $\approx \log_2 n$ in a random graph
from $\G_{n,1/2}$, no such algorithm is known to do better.
The Erd\H{o}s-R\'enyi model fails to
distinguish between different efficient heuristics for the Maximum
Clique problem.

\subsection{Planted Graph Models}\label{ss:planted}

Chapters~5 and~6 study deterministic models of data in which the
optimal solution to an optimization problem must be ``clearly
optimal'' in some sense, with the motivation of zeroing in on the
instances with a ``meaningful'' solution (such as an informative
clustering of data points).  {\em Planted graph models} implement the
same stability idea in the context of random graphs, by positing
probability distributions over inputs which generate (with high
probability) graphs in which an optimal solution ``sticks out.''  The
goal is then to devise a polynomial-time algorithm that recovers the
optimal solution with high probability, under the weakest-possible
assumptions on the input distribution.  Unlike an oblivious random
model such as the Erd\H{o}s-R\'enyi model, planted models are generally
defined with a particular computational problem in mind.

Algorithms for planted models generally fall into three categories,
listed roughly in order of increasing complexity and power.
\begin{enumerate}

\item {\em Combinatorial approaches.}  We leave the term ``combinatorial''
  safely undefined, but basically it refers to algorithms that work
  directly with the graph, rather than resorting to any continuous
  methods.   For example, an algorithm that looks only at vertex
  degrees,   subgraphs, shortest paths, etc., would be considered
  combinatorial.  

\item {\em Spectral algorithms.}  Here ``spectral'' means an
  algorithm that
  computes and uses the eigenvectors of a suitable
  matrix derived from the input graph.  Spectral algorithms often
  achieve optimal recovery guarantees for planted models.

\item {\em Semidefinite programming (SDP).} Algorithms that use
  semidefinite programming have proved useful for extending guarantees
  for spectral algorithms in planted models to hold also in semi-random
  models (see Chapters~9 and 10).

\end{enumerate}

\vspace{-.5\baselineskip}
\paragraph{Example: Planted bisection.}
In the planted bisection problem, a graph is generated according to
the following random process (for a fixed vertex set $V$, with $|V|$
even, and parameters $p,q \in [0,1]$):
\begin{enumerate}

\item Choose a partition $(S,T)$ of $V$ with $|S|=|T|$ uniformly at
  random.

\item Independently for each pair $(i,j)$ of vertices inside the same
  cluster ($S$ or $T$), include the edge $(i,j)$ with probability $p$.

\item Independently for each pair $(i,j)$ of vertices in different clusters,
include the edge $(i,j)$ with probability $q$.\footnote{This model is
  a special case of the stochastic block model studied in Chapter~10.}

\end{enumerate}
Thus the expected edge density inside the clusters is $p$, and between
the clusters is~$q$.

The difficulty of recovering the planted bisection~$(S,T)$ clearly
depends on the gap between $p$ and $q$.  The problem is impossible if
$p=q$ and trivial if~$p=1$ and~$q=0$.
Thus the key question in
this model is: how big does the gap $p-q$ need to be before exact
recovery is possible in polynomial time (with high probability)?  

When $p$, $q$, and $p-q$ are bounded below by a constant independent
of~$n$, the problem is easily solved by combinatorial approaches
(Exercise~\ref{exer:easy_bisection}); unfortunately,
these do not resemble algorithms that perform well in practice.

We can make the problem more difficult by allowing~$p$, $q$, and $p-q$
to go to~0 with~$n$.
Here, 
semidefinite programming-based algorithms work for an
impressively wide range of parameter values.  For example:
\begin{theorem}[\citet{ABH16,HWX16}]\label{t:sbm}
If~$p = \tfrac{\alpha \ln n}{n}$ and~$q = \tfrac{\beta \ln n}{n}$ with
$\alpha > \beta$, then:
\begin{itemize}

\item [(a)] If $\sqrt{\alpha}-\sqrt{\beta} \ge \sqrt{2}$,
there is a polynomial-time algorithm that
  recovers the planted partition~$(S,T)$ with probability~$1-o(1)$ as
  $n \rightarrow \infty$.

\item [(b)] If $\sqrt{\alpha} - \sqrt{\beta} < \sqrt{2}$, then no
  algorithm recovers the planted partition with constant probability
  as $n \rightarrow \infty$.

\end{itemize}
\end{theorem}
In this parameter regime, semidefinite programming algorithms
provably achieve
information-theoretically optimal recovery guarantees.  
Thus, switching from the $p,q,p-q=\Omega(1)$ parameter regime to the
$p,q,p-q=o(1)$ regime is valuable not because we literally
believe that the latter is more faithful to ``real-world''
instances, but rather because it encourages better algorithm design.

\paragraph{Example: Planted clique.}
The {\em planted clique} problem with parameters~$k$ and $n$ concerns
the following distribution over graphs.
\begin{enumerate}

\item Fix a vertex set $V$ with $n$ vertices.  Sample a graph from~$\G_{n,1/2}$:
Independently for each pair $(i,j)$ of vertices,
include the edge $(i,j)$ with probability $\tfrac{1}{2}$.

\item Choose a random subset $Q \subseteq V$ of $k$ vertices.

\item Add all remaining edges between pairs of vertices in $Q$.

\end{enumerate}
Once~$k$ is significantly bigger than~$\approx 2 \log_2 n$, the likely
size of a maximum clique in a random graph from $\G_{n,1/2}$,
the planted clique $Q$ is with high probability the maximum clique of
the graph.  How big does $k$ need to be before it becomes visible to
a polynomial-time algorithm?  

When~$k = \Omega(\sqrt{n \log n})$, the problem is trivial, with the
$k$ highest-degree vertices constituting the planted clique~$Q$.
To see why this
is true, think first about the sampled Erd\H{o}s-R\'enyi random
graph, before the clique~$Q$ is planted.  The expected degree of each
vertex is $\approx n/2$, with standard deviation $\approx \sqrt{n}/2$.
Textbook large deviation inequalities show that, with high
probability, the degree of every vertex is within
$\approx \sqrt{\ln n}$ standard deviations of its expectation
(Figure~\ref{f:catapult}).  Planting a clique~$Q$ of size
$a \sqrt{n \log n}$, for a sufficiently large constant~$a$, then
boosts the degrees of all of the clique vertices enough that they
catapult past the degrees of all of the non-clique vertices.

\begin{figure}
\begin{center}
\includegraphics[width=.75\textwidth]{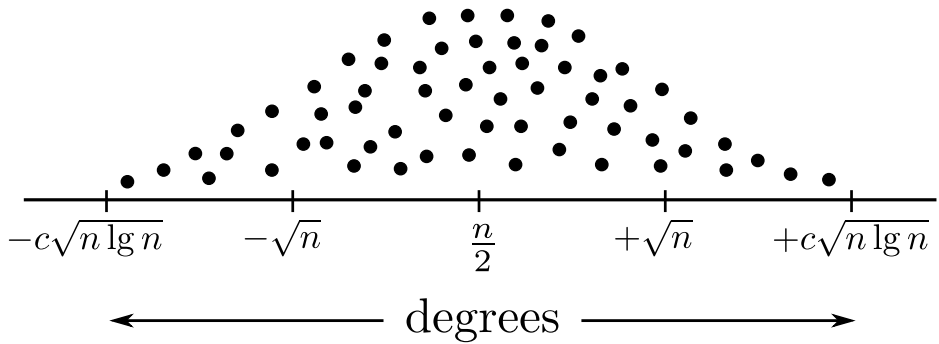}
\caption{Degree distribution of an Erd\H{o}s-R\'enyi random
graph with edge density $\tfrac{1}{2}$,
before planting
the $k$-clique~$Q$. 
If $k = \Omega(\sqrt{n \lg n})$, then the planted clique will consist
  of the $k$ vertices with the highest degrees.}\label{f:catapult} 
\end{center}
\end{figure}

The ``highest degrees'' algorithm is not very useful in practice.
What went wrong?  The same thing that often goes wrong with pure
average-case analysis---the solution is brittle and overly tailored to
a specific distributional assumption.  How can we change the input
model to encourage the design of algorithms with more robust
guarantees?  

One idea is to mimic what worked well for the planted bisection problem,
and to study a more difficult parameter regime that forces us to
develop more useful algorithms.  For the planted clique problem, there
are non-trivial algorithms, including spectral algorithms, that
recover the planted clique~$Q$ with high probability provided $k=\Omega(\sqrt{n})$ (see the Notes).

\subsection{Discussion}

There is a happy ending to the study of both the planted
bisection and planted clique problems: with the right choice
of parameter regimes, these models drive us toward non-trivial
algorithms that might plausibly be useful starting points for the
design of practical algorithms.
Still, both results seem to emerge from ``threading the needle'' in
the parameter space.
Could there be a
better alternative, in the form of input models that explicitly
encourage the design of robustly good algorithms?

\section{Robust Distributional Analysis}\label{s:rda}

Many of the remaining chapters in this book pursue different hybrids
of worst- and average-case analysis, in search of a ``sweet spot'' for
algorithm analysis that both encourages robustly good algorithms (like
in worst-case analysis) and allows for strong provable guarantees
(like in average-case analysis).  Most of these models assume that
there is in fact a probability distribution over inputs (as in
average-case analysis), but that this distribution is a priori {\em
  unknown} to an algorithm.  The goal is then to design algorithms
that work well no matter what the input distribution is (perhaps with
some restrictions on the class of possible distributions).  Indeed,
several of the average-case guarantees in this chapter can be viewed
as applying simultaneously (i.e., in the worst case) across a
restricted but still infinite family of input distributions:
\begin{itemize}

\item The $\tfrac{1}{2}$-approximation in the prophet inequality
  (Theorem~\ref{t:pi}) for a threshold-$t$ rule applies simultaneously
  to all distribution sequences $D_1,D_2,\ldots,D_n$ such that
  $\prob[\prizes \sim \gdists]{\max_i v_i \ge t} = \tfrac{1}{2}$
(e.g., all possible reorderings of one such sequence).

\item The guarantees for our algorithms for the bin packing
  (Theorem~\ref{t:f80}), convex hull (Theorem~\ref{t:ch}), and
  Euclidean TSP (Theorem~\ref{t:k77}) problems 
hold  more generally for 
all input distributions that are sufficiently close to uniform.

\end{itemize}
The general research agenda in robust distributional analysis is to
prove approximate optimality guarantees for algorithms for as many
different computational problems and as rich a class of input
distributions as possible.  Work in the area can be divided into two
categories, both well represented in this book, depending on whether
an algorithm observes one or many
samples from the unknown input distribution.  We conclude this chapter
with an overview of what's to come.

\subsection{Simultaneous Near-Optimality}

In {\em single-sample} models, an algorithm is designed with knowledge
only of a class~$\D$ of possible input distributions, and receives
only a single input drawn from an unknown and adversarially chosen
distribution from~$\D$.  In these models, the algorithm cannot hope to
learn anything non-trivial about the input distribution.
Instead, the goal is to design an algorithm that, for
every input distribution $D \in \D$, has expected performance close to
that of the optimal algorithm specifically tailored for~$D$.  Examples
include:
\begin{itemize}

\item The semi-random models in Chapters~9--11 and~17 and the smoothed
  analysis models in Chapters~13--15 and~19.  In these models, nature
  and an adversary collaborate to produce an input, and each fixed
  adversary strategy induces a particular input distribution.  
Performing well with respect to the adversary in these models is
equivalent to performing well simultaneously across all of the induced
input distributions.

\item The effectiveness of simple hash functions with pseudorandom
  data (Chapter~26).  The main result in this chapter is a guarantee for
  universal hashing that holds simultaneously across all data
  distributions with sufficient entropy.

\item Prior-independent auctions (Chapter~27), which are auctions that
  achieve near-optimal expected revenue simultaneously across a wide
  class of valuation distributions.

\end{itemize}

\subsection{Learning a Near-Optimal Solution}

In {\em multi-sample} models, an algorithm observes multiple samples
from an unknown input distribution~$D \in \D$, and the goal
is to efficiently identify a near-optimal algorithm for~$D$ from as
few samples as possible.  Examples include:
\begin{itemize}

\item Self-improving algorithms (Chapter~12) and data-driven
  algorithm design (Chapter~29), in which the goal is to design an
  algorithm that, when presented with independent
  samples from an unknown input distribution, quickly converges to an
approximately best-in-class algorithm for that distribution.

\item Supervised learning (Chapters~16 and 22), in which the goal is to
  identify the expected loss-minimizing hypothesis (from a given
  hypothesis   class) for an unknown data distribution given samples from that
  distribution.

\item Distribution testing (Chapter~23), in which the goal is to make
  accurate inferences about an unknown distribution from a limited
  number of samples.

\end{itemize}

\section{Notes}

The prophet inequality (Theorem~\ref{t:pi}) is due to \citet{SC84}.
The pros and cons of threshold rules versus optimal online algorithms
are discussed also by \citet{hartline}.
QuickSort and its original analysis are due to \citet{hoare}.
The~$(1-\alpha)^{-2}$ bound for linear probing with load~$\alpha$ and
random data, as well as the corresponding quote in
Section~\ref{ss:knuth}, are in \citet{knuth}.
A good (if outdated) entry point to the literature on bin packing is
\citet{CGJ96}.  The lower bound for the FFD algorithm in
Exercise~\ref{exer:ffd} is from \citet{J+74}.  The first upper bound
of the form $\tfrac{11}{9} \cdot OPT + O(1)$ for
the number of bins used by the FFD algorithm, where~$OPT$ denotes the
minimum-possible number of bins, is due to \citet{J73}.  The exact
worst-case bound for FFD was pinned down recently by \citet{DLHY13}.
The average-case guarantee in Theorem~\ref{t:f80} is a variation on
one by \citet{F80}, who proved that the expected difference between the
number of bins used by FFD and an optimal solution is~$O(n^{2/3})$.  A
more sophisticated argument gives a tight bound of~$\Theta(n^{1/2})$
on this expectation \citep{C+91}.

The linear expected time algorithm for 2D convex hulls
(Theorem~\ref{t:ch}) is by \cite{BS78}.
Lemma~\ref{l:ch} was first proved by \citet{RS63}; the proof outlined
here follows~\citet{sariel}.  Exercise~\ref{exer:graham}
is solved by \citet{A79}.  The asymptotic optimality of the Stitch
algorithm for Euclidean TSP (Theorem~\ref{t:k77}) is due to
\citet{K77}, who also gave an alternative solution based on the
adaptive dissections mentioned in Section~\ref{ss:disc}.  
A good general reference for this topic is \citet{KS85}.
The worst-case approximation schemes mentioned in
Section~\ref{ss:disc} are due to \citet{A98} and \citet{M99}.

The Erd\H{o}s-R\'enyi random graph model is from \citet{ER60}.
The size of the maximum clique in a random graph drawn
from~$\G_{n,1/2}$ was characterized by \citet{M76}; with high
probability it is either $k$ or $k+1$, where $k$ is an integer roughly
equal to $2 \log_2 n$.  \citet{GM75} proved that the greedy algorithm
finds, with high probability, a clique of size roughly $\log_2 n$ in a
random graph from $\G_{n,1/2}$.
The planted bisection model described here was proposed by
\citet{bui} and is also a special case of the stochastic block model 
defined by \cite{HLL83}.
Part~(b) of Theorem~\ref{t:sbm} and a weaker version of part~(a) were
proved by \citet{ABH16}; the stated version of part~(a) is due to
\citet{HWX16}.
The planted clique model was suggested by \citet{J92}.
\citet{kucera} noted that the ``top-$k$ degrees'' algorithm works with
high probability when $k = \Omega(\sqrt{n \log n})$.  The first
polynomial-time algorithm for the planted clique problem
with~$k=O(\sqrt{n})$ was the spectral algorithm
of \citet{AKS98}.  \citet{sos} supplied evidence, in the form of
a sum-of-squares lower bound, that the planted clique problem is
intractable when~$k=o(\sqrt{n})$.

The versions of the Chernoff bound stated in
Exercises~\ref{exer:f80a}(a) and~\ref{exer:k77b}
can be found, for example, in \citet{MU17}.

\section*{Acknowledgments}

I thank Anupam Gupta, C.\ Seshadhri, and Sahil Singla for helpful
comments on a preliminary draft of this chapter.

\section*{Exercises}

\begin{enumerate}

\item \label{exer:prophetmass}
Extend the prophet inequality (Theorem~\ref{t:pi}) to the case in which
there is no threshold~$t$ with $q(t) = \tfrac{1}{2}$, where $q(t)$ is
the probability that no prize meets the threshold.  

\vspace{.25\baselineskip}
\noindent 
[Hint: Define $t$ such that $\Pr[\pi_i > t \mbox{ for all $i$}] \le
\tfrac{1}{2} \le \Pr[\pi_i \ge t \mbox{ for all $i$}]$.  Show that at
least one of the two corresponding strategies---either taking the
first prize with value at least~$t$, or the first with value
exceeding~$t$---satisfies the requirement.]

\item \label{exer:opt_vs_prophet}
The prophet inequality (Theorem~\ref{t:pi}) provides an approximation
guarantee of~$\tfrac{1}{2}$ relative to the expected prize value obtained
by a prophet, which is at least (and possibly more than) the expected
prize value obtained by an optimal online algorithm.  
Show by examples that the latter quantity can range from 50\%
to 100\% of the former.

\item \label{exer:ffd} Prove that for a bin packing instance 
consisting of~6 items with size~$\tfrac{1}{2}+\eps$, 6 items with size
$\tfrac{1}{4}+2\eps$, 6 
jobs with size $\tfrac{1}{4}+\eps$, and 12 items with
size~$\tfrac{1}{4}-2\eps$, the first-fit decreasing algorithm uses~11
bins and an optimal solution uses~9 bins.

\item \label{exer:f80a}
This exercise and the next outline a proof of Theorem~\ref{t:f80}.
Divide the
interval~$[0,1]$ evenly into $n^{1/4}$ intervals, with~$I_j$ denoting
the subinterval $[\tfrac{j-1}{n^{1/4}},\tfrac{j}{n^{1/4}}]$
for~$j=1,2,\ldots,n^{1/4}$.  Let~$P_j$ denote the items with size
in~$I_j$.
\begin{exlist}

\item 
One version of the Chernoff bound states that, for
every sequence $X_1,X_2,\ldots,X_n$ of Bernoulli (0-1) random
variables with means~$p_1,p_2,\ldots,p_n$ and every $\delta \in
(0,1)$,
\[
\prob{|X-\mu| \ge \delta \mu} \le 2e^{-\mu \delta^2/3},
\]
where~$X$ and~$\mu$ denote $\sum_{i=1}^n X_i$ and $\sum_{i=1}^n p_i$,
respectively.
Use this bound to prove that 
\begin{equation}\label{eq:pop}
|P_j| \in \left[ n^{3/4} - \sqrt{n}, n^{3/4} + \sqrt{n} \right] \text{
  for all $j=1,2,\ldots,n^{1/4}$}
\end{equation}
with probability~$1-o(1)$ as $n \rightarrow \infty$.

\item Assuming~\eqref{eq:pop}, prove that the sum~$\sum_{i=1}^n s_i$
  is at least $\tfrac{1}{2}n - c_1 n^{3/4}$ for some constant $c > 0$.
What does this imply about the number of bins used by an optimal
solution?

\item Assuming~\eqref{eq:pop}, prove that in the third step of the TM
  algorithm, every pair of items~$i$ and~$k-i+1$ fits in a single
  bin.

\item Conclude that there is a constant $c_2 > 0$ such that, when property~\eqref{eq:pop} holds, the TM
  algorithm uses at most $\tfrac{1}{2}n + c_2 n^{3/4} = (1+o(1)) \cdot
  OPT$ bins, where~$OPT$ denotes the number of bins used by an optimal
  solution.

\end{exlist}

\item \label{exer:f80b} Prove Lemma~\ref{l:f80}.

\item \label{exer:graham}
Give an algorithm that, given a set~$S$ of~$n$ points from the square
sorted by $x$-coordinate, computes the convex hull of~$S$ in~$O(n)$
time.

\vspace{.25\baselineskip}
\noindent
[Hint: compute the lower and upper parts of the convex hull separately.]

\item \label{exer:ch}
Prove that the convex hull of~$n$ points drawn independently and
uniformly at random from the unit square occupies a~$1-O(\tfrac{\log
  n}{n})$ fraction of the square.

\item \label{exer:k77} Prove Lemma~\ref{l:k77a}.

\vspace{.25\baselineskip}
\noindent
[Hint: Chop the unit square evenly into $n$ subsquares of side
length~$n^{-1/2}$, and each subsquare further into 9 mini-squares of
side length~$\tfrac{1}{3} \cdot n^{-1/2}$.  For a given subsquare, what is the
probability that the input includes one point from its center
mini-square and none from the other~8 mini-squares?]

\item \label{exer:k77b}
Another variation of the Chernoff bound states that, for
every sequence $X_1,X_2,\ldots,X_n$ of Bernoulli (0-1) random
variables with means~$p_1,p_2,\ldots,p_n$ and every $t \ge 6\mu$,
\[
\prob{X \ge t} \le 2^{-t},
\]
where~$X$ and~$\mu$ denote $\sum_{i=1}^n X_i$ and $\sum_{i=1}^n p_i$,
respectively.  Use this bound to prove Lemma~\ref{l:k77b}.

\item \label{exer:k77c} Prove Lemma~\ref{l:k77c}.

\item \label{exer:bisection} Use the Chernoff bound from
  Exercise~\ref{exer:f80a}(a) to prove that, with probability
approaching~1 as $n \rightarrow \infty$, every bisection of a random
graph from $\G_{n,p}$ has $(1\pm o(1)) \cdot \tfrac{n^2}{8}$ crossing
edges.

\item \label{exer:easy_bisection}
Consider the planted bisection problem
with parameters $p = c_1$ and $q
= p - c_2$ for constants $c_1,c_2 > 0$.
Consider the following simple combinatorial algorithm for recovering a
planted bisection:
\begin{itemize}

\item Choose a vertex $v$ arbitrarily.

\item Let $A$ denote the $\tfrac{n}{2}$ vertices that have the
  fewest common neighbors with~$v$.

\item Let $B$ denote the rest of the vertices (including $v$) and
  return $(A,B)$.

\end{itemize}
Prove that, with high probability over the random choice of $G$
(approaching~1 as $n \rightarrow
\infty$), this algorithm exactly recovers the planted bisection.

\vspace{.25\baselineskip}
\noindent
[Hint: compute the expected number of common neighbors for pairs of
vertices on the same and on different sides of the planted partition.
Use the Chernoff bound.]

\item 
Consider the planted clique problem (Section~\ref{ss:planted}) with
planted clique size $k \ge c \log_2 n$ for a sufficiently large
constant $c$.  Design an algorithm that runs in $n^{O(\log n)}$ time
and, with probability $1-o(1)$ as $n
\rightarrow \infty$, recovers the planted clique.

\end{enumerate}


\begin{thebibliography}{}

\bibitem[\protect\citeauthoryear{Abbe, Bandeira, and Hall}{Abbe
  et~al.}{2016}]{ABH16}
Abbe, E., A.~S. Bandeira, and G.~Hall (2016).
\newblock Exact recovery in the stochastic block model.
\newblock {\em IEEE Transactions on Information Theory\/}~{\em 62\/}(1),
  471--487.

\bibitem[\protect\citeauthoryear{Alon, Krivelevich, and Sudakov}{Alon
  et~al.}{1998}]{AKS98}
Alon, N., M.~Krivelevich, and B.~Sudakov (1998).
\newblock Finding a large hidden clique in a random graph.
\newblock {\em Random Structures {\&} Algorithms\/}~{\em 13\/}(3-4), 457--466.

\bibitem[\protect\citeauthoryear{Andrews}{Andrews}{1979}]{A79}
Andrews, A.~M. (1979).
\newblock Another efficient algorithm for convex hulls in two dimensions.
\newblock {\em Information Processing Letters\/}~{\em 9\/}(5), 216--219.

\bibitem[\protect\citeauthoryear{Arora}{Arora}{1998}]{A98}
Arora, S. (1998).
\newblock Polynomial time approximation schemes for euclidean traveling
  salesman and other geometric problems.
\newblock {\em Journal of the ACM\/}~{\em 45\/}(5), 753--782.

\bibitem[\protect\citeauthoryear{Barak, Hopkins, Kelner, Kothari, Moitra, and
  Potechin}{Barak et~al.}{2016}]{sos}
Barak, B., S.~B. Hopkins, J.~A. Kelner, P.~Kothari, A.~Moitra, and A.~Potechin
  (2016).
\newblock A nearly tight sum-of-squares lower bound for the planted clique
  problem.
\newblock In {\em Proceedings of the 57th Annual IEEE Symposium on Foundations
  of Computer Science {(FOCS)}}, pp.\  428--437.

\bibitem[\protect\citeauthoryear{Bentley and Shamos}{Bentley and
  Shamos}{1978}]{BS78}
Bentley, J.~L. and M.~I. Shamos (1978).
\newblock Divide and conquer for linear expected time.
\newblock {\em Information Processing Letters\/}~{\em 7\/}(2), 87--91.

\bibitem[\protect\citeauthoryear{Bui, Chaudhuri, Leighton, and Sipser}{Bui
  et~al.}{1987}]{bui}
Bui, T.~N., S.~Chaudhuri, F.~T. Leighton, and M.~Sipser (1987).
\newblock Graph bisection algorithms with good average case behavior.
\newblock {\em Combinatorica\/}~{\em 7\/}(2), 171--191.

\bibitem[\protect\citeauthoryear{{Coffman, Jr.}, Courcoubetis, Garey, Johnson,
  McGeoch, Shor, Weber, and Yannakakis}{{Coffman, Jr.} et~al.}{1991}]{C+91}
{Coffman, Jr.}, E.~G., C.~Courcoubetis, M.~R. Garey, D.~S. Johnson, L.~A.
  McGeoch, P.~W. Shor, R.~R. Weber, and M.~Yannakakis (1991).
\newblock Fundamental discrepancies between average-case analyses under
  discrete and continuous distributions: A bin packing case study.
\newblock In {\em Proceedings of the 23rd Annual ACM Symposium on Theory of
  Computing (STOC)}, pp.\  230--240.

\bibitem[\protect\citeauthoryear{{Coffman, Jr.}, Garey, and Johnson}{{Coffman,
  Jr.} et~al.}{1996}]{CGJ96}
{Coffman, Jr.}, E.~G., M.~R. Garey, and D.~S. Johnson (1996).
\newblock Approximation algorithms for bin packing: A survey.
\newblock In D.~Hochbaum (Ed.), {\em Approximation Algorithms for NP-Hard
  Problems}, Chapter~2, pp.\  46--93. PWS.

\bibitem[\protect\citeauthoryear{D{\'o}sa, Li, Hanc, and Tuza}{D{\'o}sa
  et~al.}{2013}]{DLHY13}
D{\'o}sa, G., R.~Li, X.~Hanc, and Z.~Tuza (2013).
\newblock Tight absolute bound for first fit decreasing bin-packing: {$FFD(L)
  \le 11/9 OPT(L)+6/9$}.
\newblock {\em Theoretical Computer Science\/}~{\em 510}, 13--61.

\bibitem[\protect\citeauthoryear{Erd{\H{o}}s and R{\'e}nyi}{Erd{\H{o}}s and
  R{\'e}nyi}{1960}]{ER60}
Erd{\H{o}}s, P. and A.~R{\'e}nyi (1960).
\newblock On the evolution of random graphs.
\newblock {\em Publ. Math. Inst. Hungar. Acad. Sci.\/}~{\em 5}, 17--61.

\bibitem[\protect\citeauthoryear{Frederickson}{Frederickson}{1980}]{F80}
Frederickson, G.~N. (1980).
\newblock Probabilistic analysis for simple one- and two-dimensional bin
  packing algorithms.
\newblock {\em Information Processing Letters\/}~{\em 11\/}(4-5), 156--161.

\bibitem[\protect\citeauthoryear{Grimmett and {McDiarmid}}{Grimmett and
  {McDiarmid}}{1975}]{GM75}
Grimmett, G. and C.~J.~H. {McDiarmid} (1975).
\newblock On colouring random graphs.
\newblock {\em Mathematical Proceedings of the Cambridge Philosophical
  Society\/}~{\em 77}, 313--324.

\bibitem[\protect\citeauthoryear{Hajek, Wu, and Xu}{Hajek et~al.}{2016}]{HWX16}
Hajek, B., Y.~Wu, and J.~Xu (2016).
\newblock Achieving exact cluster recovery threshold via semidefinite
  programming: Extensions.
\newblock {\em IEEE Transactions on Information Theory\/}~{\em 62\/}(10),
  5918--5937.

\bibitem[\protect\citeauthoryear{Har-Peled}{Har-Peled}{1998}]{sariel}
Har-Peled, S. (1998).
\newblock On the expected complexity of random convex hulls.
\newblock Technical Report 330/98, School of Mathematical Sciences, Tel Aviv
  University.

\bibitem[\protect\citeauthoryear{Hartline}{Hartline}{2017}]{hartline}
Hartline, J.~D. (2017).
\newblock Mechanism design and approximation.
\newblock Book in preparation.

\bibitem[\protect\citeauthoryear{Hoare}{Hoare}{1962}]{hoare}
Hoare, C. A.~R. (1962).
\newblock Quicksort.
\newblock {\em The Computer Journal\/}~{\em 5\/}(1), 10--15.

\bibitem[\protect\citeauthoryear{Holland, Lasket, and Leinhardt}{Holland
  et~al.}{1983}]{HLL83}
Holland, P.~W., K.~Lasket, and S.~Leinhardt (1983).
\newblock Stochastic blockmodels: First steps.
\newblock {\em Social Networks\/}~{\em 5\/}(2), 109--137.

\bibitem[\protect\citeauthoryear{Jerrum}{Jerrum}{1992}]{J92}
Jerrum, M. (1992).
\newblock Large cliques elude the {M}etropolis process.
\newblock {\em Random Structures and Algorithms\/}~{\em 3\/}(4), 347--359.

\bibitem[\protect\citeauthoryear{Johnson}{Johnson}{1973}]{J73}
Johnson, D.~S. (1973).
\newblock {\em Near-Optimal Bin Packing Algorithms}.
\newblock Ph.\ D. thesis, MIT.

\bibitem[\protect\citeauthoryear{Johnson, Demers, Ullman, Garey, and
  Graham}{Johnson et~al.}{1974}]{J+74}
Johnson, D.~S., A.~Demers, J.~D. Ullman, M.~R. Garey, and R.~L. Graham (1974).
\newblock Worst-case performance bounds for simple one-dimensional packing
  algorithms.
\newblock {\em SIAM Journal on Computing\/}~{\em 3\/}(4), 299--325.

\bibitem[\protect\citeauthoryear{Karp}{Karp}{1977}]{K77}
Karp, R.~M. (1977).
\newblock Probabilistic analysis of partitioning algorithms for the
  traveling-salesman problem in the plane.
\newblock {\em Mathematics of Operations Research\/}~{\em 2\/}(3), 209--224.

\bibitem[\protect\citeauthoryear{Karp and Steele}{Karp and Steele}{1985}]{KS85}
Karp, R.~M. and J.~M. Steele (1985).
\newblock Probabilistic analysis of heuristics.
\newblock In E.~L. Lawler, J.~K. Lenstra, A.~H.~G. {Rinnooy Kan}, and D.~B.
  Shmoys (Eds.), {\em The Traveling Salesman Problem}, Chapter~6, pp.\
  181--205. John Wiley {\&} Sons.

\bibitem[\protect\citeauthoryear{Knuth}{Knuth}{1998}]{knuth}
Knuth, D.~E. (1998).
\newblock {\em The Art of Computer Programming: Sorting and Searching},
  Volume~3.
\newblock Addison-Wesley.
\newblock Second edition.

\bibitem[\protect\citeauthoryear{Kucera}{Kucera}{1995}]{kucera}
Kucera, L. (1995).
\newblock Expected complexity of graph partitioning problems.
\newblock {\em Discrete Applied Mathematics\/}~{\em 57\/}(2-3), 193--212.

\bibitem[\protect\citeauthoryear{Matula}{Matula}{1976}]{M76}
Matula, D.~W. (1976).
\newblock The largest clique size in a random graph.
\newblock Technical Report 7608, Department of Computer Science, Southern
  Methodist University.

\bibitem[\protect\citeauthoryear{Mitchell}{Mitchell}{1999}]{M99}
Mitchell, J. S.~B. (1999).
\newblock Guillotine subdivisions approximate polygonal subdivisions: A simple
  polynomial-time approximation scheme for geometric tsp, k-mst, and related
  problems.
\newblock {\em SIAM Journal on Computing\/}~{\em 28\/}(4), 1298--1309.

\bibitem[\protect\citeauthoryear{Mitzenmacher and Upfal}{Mitzenmacher and
  Upfal}{2017}]{MU17}
Mitzenmacher, M. and E.~Upfal (2017).
\newblock {\em Probability and Computing}.
\newblock Cambridge.
\newblock Second edition.

\bibitem[\protect\citeauthoryear{R{\'enyi} and Sulanke}{R{\'enyi} and
  Sulanke}{1963}]{RS63}
R{\'enyi}, A. and R.~Sulanke (1963).
\newblock {\"U}ber die konvexe h{\"u}lle von {$n$} zug{\"a}llig gew{\"a}hlten
  punkten.
\newblock {\em Zeitschrift f{\"u}r Wahrscheinlichkeitstheorie und Verwandte
  Gebiete\/}~{\em 2}, 75--84.

\bibitem[\protect\citeauthoryear{Roughgarden}{Roughgarden}{2020}]{bwca}
Roughgarden, T. (Ed.) (2020).
\newblock {\em Beyond the Worst-Case Analysis of Algorithms}.
\newblock Cambridge University Press.

\bibitem[\protect\citeauthoryear{Samuel-{C}ahn}{Samuel-{C}ahn}{1984}]{SC84}
Samuel-{C}ahn, E. (1984).
\newblock Comparison of threshold stop rules and maximum for independent
  nonnegative random variables.
\newblock {\em Annals of Probability\/}~{\em 12\/}(4), 1213--1216.

\end{thebibliography}
\end{document}